\documentclass[conference]{IEEEtran}

\usepackage{amsmath,amssymb,amsfonts,mathrsfs,bm}
\usepackage{amstext}
\usepackage{upgreek}
\usepackage{multicol}
\usepackage{indentfirst}
\usepackage{graphicx}
\usepackage{paralist}
\usepackage{multirow}
\usepackage{tabularx}
\usepackage{cite}
\usepackage{citesort}
\usepackage{booktabs}
\usepackage{subfigure}
\usepackage{color}
\IEEEoverridecommandlockouts
\allowdisplaybreaks

\newtheorem{corollary}{\bf Corollary}
\newtheorem{proposition}{\bf Proposition}

\newtheorem{proof}{Proof}
\usepackage{accents}
\newcommand{\ubar}[1]{\underaccent{\bar}{#1}}

\usepackage{breakurl}
\begin{document}
\title{Optimal Content Placement for Offloading in Cache-enabled Heterogeneous Wireless Networks}
\author{
	\IEEEauthorblockN{{\large Dong Liu and Chenyang Yang}}\\
	\vspace{-3mm}
	\IEEEauthorblockA{Beihang University, Beijing, China\\
		Email: \{dliu, cyyang\}@buaa.edu.cn}}
	\maketitle 
\begin{abstract}
Caching at base stations (BSs) is a promising way to offload traffic and eliminate backhaul bottleneck in heterogeneous networks (HetNets). In this paper, we investigate the optimal content placement maximizing the successful offloading probability in a cache-enabled HetNet where a tier of multi-antenna macro BSs (MBSs) is overlaid with a tier of helpers with caches. Based on probabilistic caching framework, we resort to stochastic geometry theory to derive the closed-form successful offloading probability and formulate the caching  probability optimization problem, which is not concave in general. In two extreme cases with high and low user-to-helper density ratios, we obtain the optimal caching probability and analyze the impacts of BS density and transmit power of the two tiers and the signal-to-interference-plus-noise ratio (SINR) threshold. In general case, we obtain the optimal caching probability that maximizes the lower bound of successful offloading probability and analyze the impact of user density. Simulation and numerical results show that when the ratios of MBS-to-helper density, MBS-to-helper transmit power and user-to-helper density, and the SINR threshold are large, the optimal caching policy tends to cache the most popular files everywhere.
\end{abstract}
\section{Introduction}

Caching popular contents at the base stations (BSs) without backhaul connectivity, namely helpers, has been proposed as a promising way of alleviating the backhaul bottleneck and offloading the traffic in heterogeneous networks (HetNets) \cite{Niki13}.

Content placement is critical in reaping the benefit brought by caching. In wireless networks, when the coverage of several BSs overlaps,  a user is able to fetch contents from multiple helpers and hence cache-hit probability can be increased by caching different files among helpers. However, owing to interference and path loss, such a file diversity may lead to low signal-to-interference-plus-noise ratio (SINR), since a user may associate with a relative further BS to ``hit the cache" when the nearest BS does not cache the requested file \cite{Dong}.

In \cite{Niki13}, content placement was optimized to minimize the file download time where the interference among helpers are not considered. In \cite{song2015optimal}, the optimal caching policy was proposed to minimize the average bit error rate over fading channels. Both \cite{Niki13} and \cite{song2015optimal} assume \emph{a priori} known BS-user topology, which is not practical in mobile networks.  To reflect the uncertain connectivity between BS and user, more realistic network models based on stochastic geometry were considered recently. In \cite{Blaszczyszyn2015optimal}, a probabilistic caching policy was proposed where each BS caches files independently according to an optimized caching probability that maximizes the cache-hit probability. However, the optimal caching probability is not obtained with closed-form, which makes it hard to gain useful insights into the impacts of various system parameters. In \cite{cui2015analysis}, the optimal caching probability maximizing the successful transmission probability in a homogeneous network was obtained in closed-form when user density approaches infinity. In \cite{rao2015optimal}, content placement was optimized to maximize the traffic offloaded to helpers and cache-enabled users, but the links among helpers and users are assumed interference-free. In real-world HetNets, the interferences are complicated and have large impact on the system design and network performance. While deploying helpers is a cost-effective way for offloading,
how to place the contents is still not well understood.

In this paper, we consider a cache-enabled HetNet where a tier of macro BSs (MBSs) is overlaid with a tier of denser helpers with caches. We obtain the optimal probabilistic caching policy to maximize the successful offloading probability, and analyze the impact of critical system parameters.
The major differences from existing works are as follows.
\begin{itemize}
	\item Existing works rarely consider caching policy optimization in HetNets and the basic features of cache-enabled HetNets are not well captured and analyzed, such as BS and user densities, SINR threshold and association bias.
	\item Existing literatures never consider helper idling, which affects the optimization for the caching policy. Since the cost-effective helpers make dense deployment possible, helper idling is more appealing. Our analysis shows that by turning the helpers with no user to serve into idle mode, the  files should be cached more uniformly.
\end{itemize}
\section{System Model}
We consider a cache-enabled HetNet, where a tier of MBSs is overlaid with a tier of denser helper nodes. The MBSs are connected to the core network with high-capacity backhaul links, e.g., optical fibers, and the helper nodes are deployed without backhaul but equipped with caches.

%

The distribution of MBSs, helper nodes and users are modeled as three independent homogeneous PPPs  with density of $\lambda_1$, $\lambda_2$ and $\lambda_u$, denoted as $\Phi_{1}$, $\Phi_{2}$ and $\Phi_u$, respectively. Each MBS is equipped with $M_1 \geq 1$ antennas and each helper node is with $M_2 = 1$ antenna. Denote $k\in\{1,2\}$ as the index of the tier that a randomly chosen user in the network (called the typical user) is associated with. In the following, if not specified, BS refers to both MBS and helper. The transmit power at each BS  in each tier is denoted by $P_k$.

We assume that each user requests a file from a content catalog that contains $N_f$ files randomly, whose probability distribution is known \emph{a priori}. The files are indexed according to their popularity, ranking from the most popular (the $1$st file) to the least popular (the $N_f$th file). The probability of requesting the $f$th popular file follows Zipf distribution as
\begin{equation}
	p_f = {f^{-\delta}}\Big/\left({\textstyle\sum_{n=1}^{N_f} n^{-\delta}}\right) \label{eqn:pf}
\end{equation}
where the skew parameter $\delta$ is with typical value of 0.5 $\sim$ 1.0 \cite{breslau1999web}. For simplicity, we assume that the
files are with equal size\footnote{Files with different size can be divided into equal-size content chunks.} and the cache capacity of each helper node is $N_c$.

We consider probabilistic caching policy where each helper independently selects files to cache according to a specific probability distribution. To unify the analysis, denote $0 \leq q_{f, k} \leq 1$ as the probability that the BS of the $k$th tier caches the $f$th file. When $\mathbf q_{k} \triangleq [q_{f, k}]_{f= 1,\cdots, N_f}$ is given, each BS can determine which files should be cached by the method in \cite{Blaszczyszyn2015optimal}. For the MBS tier,  $\mathbf q_{1} = \mathbf 1$, since the MBS can be regarded as caching all the files due to the high-capacity backhaul.


Since every helper caches files independently, the distribution of the BSs in the $k$th tier that cache the $f$th file can be regarded as a thinning of the PPP $\Phi_k$ with probability $q_{f,k}$, which follows a PPP with density $q_{f, k}\lambda_k$ (denoted by $\Phi_{f,k}$). Similarly, the distribution of the BSs in the $k$th tier that do not cache the $f$th file follows PPP with density $(1-q_{f, k})\lambda_k$ (denoted by $\Phi_{f',k}$).

We consider user association based on both channel condition and content placement. Specifically, when the typical user requests the $f$th file, it associates with the BS from $\{\Phi_{f,k}\}_{k=1,2}$ that has the strongest average biased-received-power (BRP). The BRP for the $k$th tier is $P_{{\rm r}, k} = P_kB_k r^{-\alpha}$, where $B_k \geq 0$ is the association bias factor, $r$ is the BS-user distance, and $\alpha$ is the path-loss exponent.

We assume that $\lambda_u \gg \lambda_1$ such that each MBS has at least $M_1$ users to serve. Considering that the helpers  without backhaul can be densely deployed at low cost and the traffic may fluctuate among peak during off-peak times, the density of helper nodes may become comparable with the density of users and hence some helpers may have no users to serve. These inactive helpers will be turned into idle mode to avoid generating interference. Each MBS randomly selects $M_1$ users to serve at each time slot by zero-forcing beamforming with equal power allocation, while each helper randomly selects one user (if there is one) to serve at each time slot with full power. These assumptions define a typical scenario, which can capture the fundamental features of cache-enabled HetNets.

The downlink SINR at the typical user that requests the $f$th file and associates with the $k$th tier is
\begin{align}
{\gamma}_{f,k} & \! = \frac{\frac{P_k}{M_k} h_{k0} r_{k}^{-\alpha}}{\sum_{j=1}^{2} \!\!\big(\! \sum_{i \in \tilde{\Phi}_{f,j}\!\backslash  b_{k0}}\!\! P_jh_{ji} r_{ji}^{-\alpha}  +  \sum_{i \in \tilde{\Phi}_{f',j}} \!\! P_jh_{ji} r_{ji}^{-\alpha}  \big)\! +\! \sigma^2} \nonumber \\
& \triangleq \frac{\frac{P_k}{M_k} h_{k0} r_{k}^{-\alpha}}{\sum_{j=1}^{2} ( I_{f,kj} + I_{f',kj})  + \sigma^2 }
\triangleq \frac{\frac{P_k}{M_k} h_{k0} r_{k}^{-\alpha}}{I_{k}+ \sigma^2 }  \label{eqn:gamma}
\end{align}
where $h_{k0}$ is the equivalent channel (including small-scale fading and beamforming) from the associated BS $b_{k0}$ to the typical user, $r_k$ is the corresponding distance, $\tilde \Phi_{f, j}$ and $\tilde \Phi_{f', j}$ are respectively the sets of active BSs  in the $j$th tier that caching and not caching the $f$th file, $h_{ji}$ is the equivalent interference channel from the $i$th active BS in the $j$th tier to the typical user, $r_{ji}$ is the corresponding distance, and $\sigma^2$ is the noise power. We consider Rayleigh fading channels. Therefore, $h_{k0}$ follows exponential distribution with unit mean (i.e., $h_{ji} \sim \exp (1)$), and $h_{ij}$ follows gamma distribution with shape parameter $M_j$ and unit mean (i.e., $h_{ji} \sim \mathbb{G}(M_j, 1/M_j)$) \cite{adhoc}. The total interference $I_k$ consists of the interference from the BSs that caching the $f$th file (denoted by $I_{f, kj}$) and the interference from the BSs that do not cache the $f$th file (denoted by $I_{f', kj}$).

To reflect how many users in the network on average can be offloaded to the helper tier, we define the successful offloading probability as the probability that the typical user is associated with the helper tier and its downlink SINR is larger than a threshold $\gamma_{0}$,
which is expressed as \vspace{-1mm}
\begin{equation}
  P_{\rm off}  \triangleq \mathbb{P} ( \gamma > \gamma_{0} , k=2)  \label{eqn:def}
\end{equation}

\section{Optimal Caching Policy}
In this section, we find the optimal caching probability that maximizes the successful offloading probability, and analyze the impact of system settings on the optimal caching policy.

Since HetNets are usually interference-limited \cite{flexible}, it is reasonable to neglect the thermal noise, i.e., $\sigma^2 = 0$. For notational simplicity, we define the relative BS density, number of antennas, transmit power, and bias factor as $\lambda_{jk} \triangleq \lambda_j/\lambda_k$, $M_{jk} \triangleq M_j/M_k$, $P_{jk} \triangleq P_j/P_k$, and $B_{jk} \triangleq B_j/B_k$. Note that $\lambda_{kk} = M_{kk} = P_{kk} = B_{kk} = 1$.

\begin{proposition}
 The successful offloading probability of the typical user is \vspace{-2mm}
 \begin{equation}
 P_{\rm off}(\mathbf q_{2})  = \sum_{f=1}^{N_f} \frac{ p_f q_{f,2}}{ \mathcal{C}_{1,\gamma_0} + \mathcal{C}_{2,\gamma_0} p_{{\rm a},2} +  \mathcal C_{3,\gamma_0} p_{{\rm a},2} q_{f,2}  + q_{f,2} } \label{eqn:Poff} \vspace{-1mm}
 \end{equation}
 where $\mathcal{C}_{1,\gamma_0} \triangleq \lambda_{12}  (P_{12}B_{12})^{\frac{2}{\alpha}} {}_{2}F_1 \big[ -\frac{2}{\alpha}, M_1; 1-\frac{2}{\alpha}; -\frac{\gamma_0}{M_{12}B_{12}} \big]$, $\mathcal C_{2,\gamma_0} \triangleq  \Gamma (1-\frac{2}{\alpha}) \Gamma(M_2 + \frac{2}{\alpha_2}) \Gamma(M_2)^{-1}\gamma_0{}^{\frac{2}{\alpha}}$, $\mathcal{C}_{3,\gamma_0} \triangleq {}_{2}F_1  \big[ -\frac{2}{\alpha}, M_2; 1-\frac{2}{\alpha}; -\gamma_0 \big] - \mathcal C_{2,\gamma_0}- 1$, and \vspace{-1mm}
 \begin{equation}
 p_{a,2} \approx 1 -  \bigg(1 + \frac{\lambda_u}{3.5\lambda_2} \sum_{f=1}^{N_f}  \frac{p_f q_{f,2}}{ \lambda_{12} (P_{12} B_{12}){}^{2/\alpha}+ q_{f,2} }\bigg)^{-3.5} \hspace{-2mm}\label{eqn:pa} \vspace{-1mm}
 \end{equation} is the probability that a randomly chosen helper is active (i.e., has user associated with). $_{2}F_1[\cdot]$ and $\Gamma(\cdot)$ denote the Gauss hypergeometric function and Gamma function, respectively.
\end{proposition}
\begin{proof}
	See Appendix A.
\end{proof}

Then, the optimal caching probability that maximizes the successful offloading probability can be found from
\begin{subequations}
\begin{align}
\text{\em Problem 1:}\quad \max_{\mathbf q_2}~ & P_{\rm off}(\mathbf q_{2})  \nonumber \\ 
 s.t.~ & \sum_{i=1}^{N_f} q_{f,2} \leq N_c  \label{eqn:con1}  \\
& 0\leq q_{f,2} \leq 1, ~ f = 1, \cdots, N_f \label{eqn:con2}
\end{align}
\end{subequations}
where \eqref{eqn:con1} is equivalent to the cache capacity constraint (i.e., the number of cached file cannot exceed the cache capacity)  as proved in \cite{Blaszczyszyn2015optimal}, and \eqref{eqn:con2} is the probability constraint.

This problem is not concave in general, because the active probability of helper $p_{{\rm a}, 2}$ in the objective function is a complicated function of $q_{f,2}$ as shown in \eqref{eqn:pa}, which makes the global optimal solution hard to obtain.

To gain useful insights into the property of the optimal caching probability and the impact of various system settings, we first study two extreme cases where $\lambda_u / \lambda_2 \to \infty$ and $\lambda_u / \lambda_2 \to 0$, and then solve Problem 1 in general case.
\subsubsection{$\lambda_u / \lambda_2 \to \infty$}
In this case, $p_{{\rm a}, 2} \to 1$, i.e., all the helpers are active. The successful offloading probability becomes \vspace{-1mm}
\begin{equation}
 P_{\rm off}^{\infty}(\mathbf q_{2})  = \sum_{f=1}^{N_f} \frac{ p_f q_{f,2}}{ \mathcal{C}_{1,\gamma_0} + \mathcal{C}_{2,\gamma_0} + (\mathcal C_{3,\gamma_0} + 1)q_{f,2} } \label{eqn:case1}   \vspace{-1mm}
\end{equation}

It can be easily proved that the Hessian matrix of $P_{\rm off}^{\infty}(\mathbf q_2)$ is negative definite. Further considering that constraints \eqref{eqn:con1} and \eqref{eqn:con2} are linear, Problem 1 is concave when $\lambda_u/\lambda_2 \to \infty$. Then, from the Karush-Kuhn-Tucker (KKT) condition, we obtain the following proposition.
\begin{proposition}
When $\lambda_u / \lambda_2 \to \infty$, the optimal caching probability is
		\begin{equation}
		q_{f,2}^* \!= \! \left[ \frac{\sqrt{\mathcal{C}_{1,\gamma_0} + \mathcal{C}_{2,\gamma_0}} }{\sqrt{\nu} (\mathcal C_{3,\gamma_0} + 1)} \sqrt{p_f}-\frac{\mathcal{C}_{1,\gamma_0} + \mathcal{C}_{2,\gamma_0}}{\mathcal C_{3,\gamma_0} + 1} \right]_0^1 \label{eqn:opt}
		\end{equation}
		where $[x]_0^1 = \max\{\min\{x,1\},0\}$ denoting that $x$ is truncated by $0$ and $1$, and the Lagrange multiplier $\nu$ satisfying $ \sum_{f=1}^{N_f}q_{f,2}^* = N_c$ can be found by bisection searching.
\end{proposition}

As shown in \eqref{eqn:opt}, the optimal caching probability $q_{f,2}^*$ is non-increasing with $f$ since $p_f$ decreases with $f$, which coincides with the intuition that the file with higher popularity should be cached with higher probability.
For $q_{f,2}^* \in (0,1)$, considering  \eqref{eqn:pf}, the relation between caching probability and file popularity rank in \eqref{eqn:opt} obeys a shifted power law with exponent $-\delta/2$, which is very different from the noise-limited scenario considered in \cite{rao2015optimal}.

To show how BS density, transmit power of different tiers and SINR threshold  affect the optimal caching policy, we derive the following corollaries based on Proposition 2.
\begin{corollary}
For any $q_{f,2}^*$, $q_{f+1,2}^* \in (0,1)$, $q_{f,2}^* - q_{f+1,2}^*$ increases with $\lambda_{12}$ and $P_{12}$.
\end{corollary}
\begin{proof}
	See Appendix B.
\end{proof}

Corollary 1 indicates that when $\lambda_1/\lambda_2$ or $P_1/P_2$ increases, the files with higher popularity have more chances to be cached while the files with lower popularity have less chances to be cached, leading to a trend towards caching the most popular files everywhere. By contrast, when $\lambda_1/\lambda_2$ or $P_1/P_2$ decreases, the files with higher popularity have less chances to be cached while the files with lower popularity have more chances to be cached, i.e., the diversity of cached files increases.

\begin{corollary}
When $\gamma_0 \to \infty$, the optimal caching probability is $q_{1,2}^*, \cdots, q_{N_c,2}^* = 1$ and $q_{N_c+1,2}^*, \cdots, q_{N_f, 2}^* = 0$.
\end{corollary}
\begin{proof}
	See Appendix C.
\end{proof}

Corollary 2 indicates that when the SINR requirement $\gamma_0$ is high, the optimal caching policy is simply caching the most popular files everywhere.

\subsubsection{$\lambda_u/\lambda_2 \to 0$}
In this case, we have $p_{{\rm a}, 2} \to 0$.  Similar to deriving \eqref{eqn:opt}, we can obtain the optimal caching probability as
$
	q_{f2}^* =  \big[ \sqrt{\frac{ \mathcal{C}_{1,\gamma_0} }{\nu }} \sqrt{p_f}-\mathcal{C}_{1,\gamma_0}\big]_0^1 $,
	where $\nu$ satisfying $ \sum_{f=1}^{N_f}q_{f,2}^* = N_c$ can be found by bisection searching.

We can further prove that the conclusions in Corollary 1 and Corollary 2 also hold for $\lambda_u/\lambda_2 \to 0$, which are omitted  due to space limitation.

\subsubsection{General Case}
When $\lambda_u$ is comparable with $\lambda_2$, Problem 1 is not concave. Only a local optimal solution can be found, say by using inter-point method \cite{boyd2004convex}, which is of high complexity when the dimension of optimization variable $\mathbf q_2$, i.e., $N_f$, is large. Recall that it is the complicated expression of $p_{{\rm a}, 2}$ as  function of $q_{f,2}$ that makes the problem hard to solve. In the following, we first introduce a $q_{f,2}$-independent upper bound for $p_{{\rm a},2}$ (denoted by $\bar p_{{\rm a}, 2}$), which yields a lower bound of $P_{\rm off}$ (denoted by $\ubar{P}_{\rm off}$) because $P_{\rm off}$ decreases with $p_{{\rm a}, 2}$ as shown in \eqref{eqn:Poff}. Then, we solve the problem maximizing $\ubar{P}_{\rm off}$.

To obtain an upper bound of $p_{{\rm a},2}$ that does not depend on $q_{f,2}$, we first find the optimal caching probability that maximizes $p_{{\rm a},2}$, denoted as $q_{f,2}^o$. Then,  for any $q_{f,2}$ satisfying \eqref{eqn:con1} and \eqref{eqn:con2}, $ p_{{\rm a}, 2} \leq \bar p_{{\rm a}, 2}$.

From \eqref{eqn:pa}, we can see that maximizing $p_{{\rm a},2}$ is equivalent to maximizing $\sum_{f=1}^{N_f}p_f q_{f,2}  (\lambda_{12} (P_{12} B_{12}){}^{2/\alpha}+ q_{f,2} )^{-1}$, which has the same function structure as \eqref{eqn:case1}. Then, by using similar way to derive \eqref{eqn:opt}, we can obtain
$
q_{f,2}^o = \big[ \sqrt{\tfrac{\lambda_{12} (P_{12} B_{12})^{{2}/{\alpha}}}{\mu}} \sqrt{p_f} - \lambda_{12} (P_{12} B_{12})^{\frac{2}{\alpha}}  \big]_0^1
$,
where $\mu$ satisfying $\sum_{f=1}^{N_f}q_{f,2}^o = N_c$ can be found by bisection searching.

By substituting $q_{f,2}^o$ into \eqref{eqn:pa}, we can obtain $\bar{p}_{{\rm a}, 2}$. Then, by substituting $\bar{p}_{{\rm a}, 2}$ into \eqref{eqn:Poff}, we can obtain $\ubar{P}_{\rm off}$ as \vspace{-0.5mm}
\begin{equation}
\ubar{P}_{\rm off} =    \sum_{f=1}^{N_f} \frac{ p_f q_{f,2}}{ \mathcal{C}_{1,\gamma_0} + \bar{p}_{{\rm a},2}  \mathcal{C}_{2, \gamma_0} +  (\bar{p}_{{\rm a},2} \mathcal C_{3,\gamma_0} + 1) q_{f,2} }  \label{eqn:Pofflower} 
\end{equation}
Since $\bar p_{{\rm a}, 2}$ does not depend on $q_{f,2}$,  \eqref{eqn:Pofflower} has the same function structure as \eqref{eqn:case1}. Again, similar to deriving \eqref{eqn:opt}, the optimal probability that maximizes the lower bound $\ubar{P}_{\rm off}$ under the constraints \eqref{eqn:con1} and \eqref{eqn:con2} can be obtained as
	\begin{equation}
	\underline{q}_{f,2}^* \! = \!
	\left[ \frac{\sqrt{ \mathcal{C}_{1,\gamma_0} +  \bar{p}_{{\rm a},2} \mathcal{C}_{2,\gamma_0}} }{\sqrt{\nu} ( \bar{p}_{{\rm a},2} C_{3,\gamma_0} + 1)} \sqrt{p_f}-\frac{\mathcal{C}_{1,\gamma_0} +  \bar{p}_{{\rm a},2} \mathcal{C}_{2,\gamma_0}}{ \bar{p}_{{\rm a},2} C_{3, \gamma_0} + 1} \right]_0^1 \label{eqn:optlower}
	\end{equation}
where $\nu$ satisfying $\sum_{f=1}^{N_f}\ubar{q}_{f,2}^* = N_c$ can be found by bisection searching. In the next section, we show that the caching probability $\ubar{q}_{f,2}^*$ can achieve almost the same successful offloading probability as  ${q}_{f,2}^*$ found by inter-point method. Since the computation of $\ubar{q}_{f,2}^*$ only requires twice bisection searches on two scalars, i.e., $\mu$ and $\nu$, it can be obtained with much lower complexity than the inter-point method when $N_f$ is large.

With the explicit structure of $\ubar{q}_{f,2}^*$, we can analyze the impact of user density on the optimal caching policy.

\begin{corollary}
For any $\ubar{q}_{f,2}^*$, $\ubar{q}_{f+1,2}^* \in (0,1)$, $\ubar{q}_{f,2}^* - \ubar{q}_{f+1,2}^*$ increases with $\lambda_u/\lambda_2$.
\end{corollary}

\begin{proof}
	See Appendix D.
\end{proof}

Corollary 3 indicates that when the ratio of user-to-helper density increases, the files with higher popularity have more chances to be cached and \emph{vice versa}, which reflects a trend towards caching the most popular files everywhere. On the contrary, when the ratio reduces, say from $\lambda_u/\lambda_2 \to \infty$ (implies no BS idling) to finite values (implies with BS idling), $\ubar{q}_{f,2}^* - \ubar{q}_{f+1,2}^*$ decreases. This indicates that BS idling makes caching files more uniformly among the helpers optimal.
\section{Simulation and Numerical Results}
In this section, we validate our analysis via simulation and illustrate how different factors affect the optimal caching policy and corresponding successful offloading probability via numerical results. The following caching policies are considered for comparison,
\begin{enumerate}
	\item {\em Opt. (Inter-point)}: the local optimal solution of Problem 1 found by inter-point method.
	\item {\em Opt. (Lower-bound)}: the optimal solution that maximizes the lower bound $\ubar{P}_{\rm off}$ in \eqref{eqn:Pofflower}, i.e., $\ubar q_{f,2}^*$.
	\item {\em Popular}: caching the most popular files everywhere, i.e., $q_{1,2}, \cdots, q_{N_c,2} = 1$ and $q_{N_c+1,2}, \cdots, q_{N_f, 2} = 0$. This policy achieves no file diversity.
	\item {\em Uniform}: each file is cached with equal caching probability, i.e., $q_{f,2} = N_c/N_f$. This policy achieves the maximal file diversity.
\end{enumerate}

Unless otherwise specified, the following setting is used. The MBS, helper, and user densities are $\lambda_1 = 1/(250^2 \pi)$ m$^{-2}$, $\lambda_2 = 25/(250^2 \pi)$ m$^{-2}$ and $\lambda_u = 25/(250^2 \pi)$ m$^{-2}$, respectively. The path-loss exponent is $\alpha = 3.7$. The transmit power of each MBS (with $M_1 = 4$ antennas) and helper are $P_1 = 46$ dBm and $P_2 = 21$ dBm, respectively \cite{3GPP}. The bias factors for the MBS and helper tiers are $B_1 = 1$ and $B_2 = 10$ dB, respectively. The file catalog size is $N_f = 1000$ files and the SINR threshold is $\gamma_0 = -10$ dB.

\begin{figure}[!htb]
	\centering 		\vspace{-3mm}
	\includegraphics[width = 0.4\textwidth]{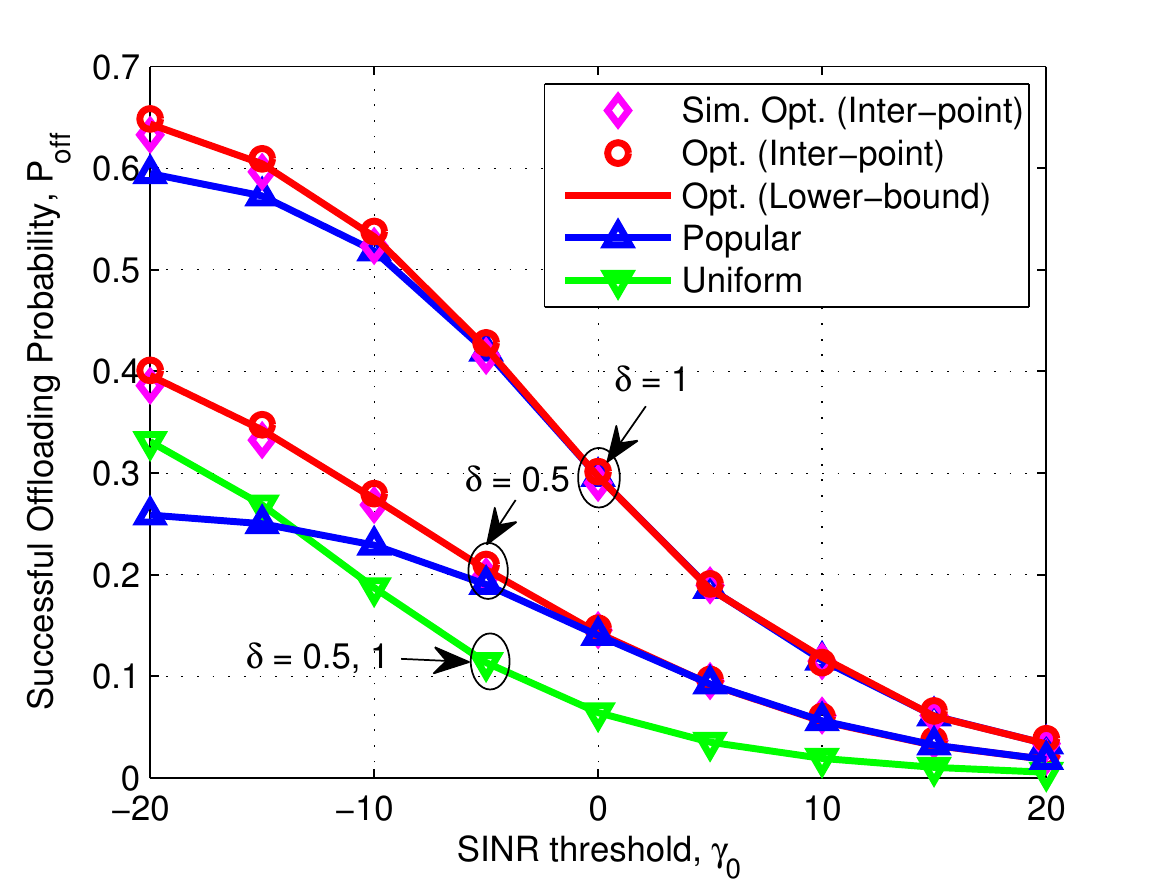}
	\vspace{-3mm}
	\caption{Impact of SINR threshold on successful offloading probability.}
	\label{fig:SINR}
		\vspace{-1mm}
\end{figure}

In Fig. \ref{fig:SINR}, we show the simulation and numerical results of successful offloading probability versus the SINR threshold. The simulation result is obtained based on $q_{f,2}^*$ from the inter-point method and then by computing $P_{\rm off} (\mathbf q_2^*)$ via Monte Carlo method considering $-95$ dBm noise power. The numerical results are computed from Proposition 1. We can see that the numerical results (with legend ``Opt. (Inter-point)'') is very close to the simulation results (with legend ``Sim. Opt. (Inter-point)'') although Proposition 1 is derived in interference-limited scenario.\footnote{For all the considered caching policies, the simulation results are overlapped with the numerical results, which are not shown for a clean figure.} Hence, in the sequel we only provided the numerical results. Comparing {``Opt. lower-bound''} with {``Opt. Inter-point''}, we can observe that the caching policy maximizing the lower bound of successful offloading probability performs almost the same  as the local  optimal solution. When the SINR threshold is high, e.g., $\gamma_0 > -5$ dB (i.e., $4$ Mbps rate for $10$ MHz bandwidth)  for $\delta = 0.5$, caching the most popular files everywhere can achieve almost the same performance as the optimized caching policies, which validates Corollary 2 although it is derived when $\gamma_0 \to \infty$ and $\lambda_u/\lambda_2 \to \infty$. Moreover, when $\delta$ increases, the gap between caching the most popular files everywhere and optimized caching policies shrinks. These results indicate  that for highly skewed demand with high rate requirement, e.g., video on demand service, simply caching the most popular files everywhere can achieve maximal successful offloading probability.

\begin{figure}[!htb]
	\centering			\vspace{-3mm}
	\includegraphics[width = 0.4\textwidth]{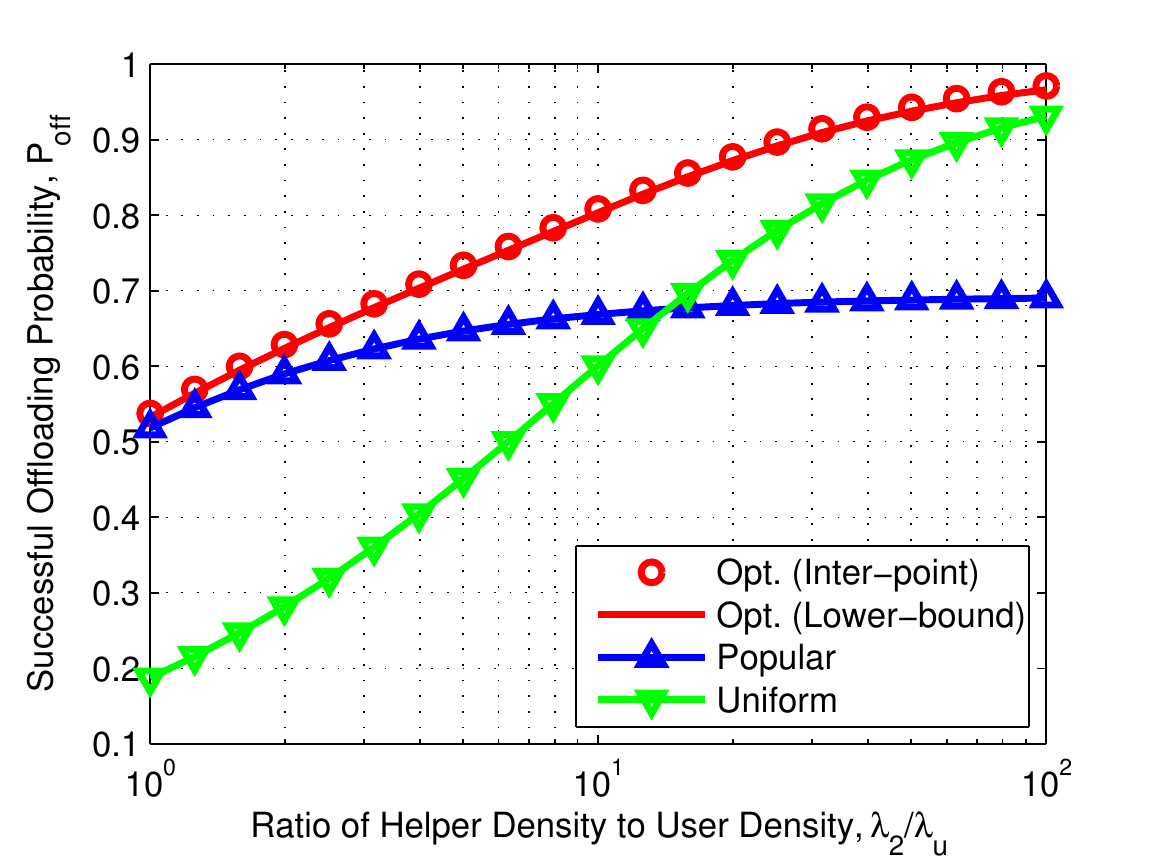}
			\vspace{-3mm}
	\caption{Impact of helper density on successful offloading probability, $\delta = 1$.}
	\label{fig:lambda}

\end{figure}

In Fig. \ref{fig:lambda}, we show the impact of helper density with given MBS and user densities. We can see that the gain of optimal caching over caching popular files everywhere shrinks when the $\lambda_2/\lambda_u$ decreases ($\lambda_2/\lambda_1$ decreases as well), which agrees with Corollary 1 and Corollary 3. When the helper density is high, the gain of optimal caching over uniform caching approaches to zero. This can be explained as follows. When $\lambda_2/\lambda_u$ is high, the SINR at the user associated with the helper tier is high since the helper is closer to the user meanwhile more helpers can be turned into idle mode leading to lower interference. As a result, increasing file diversity can improve successful offloading probability.

\begin{figure}[!htb]
	\centering		
	\includegraphics[width = 0.4\textwidth]{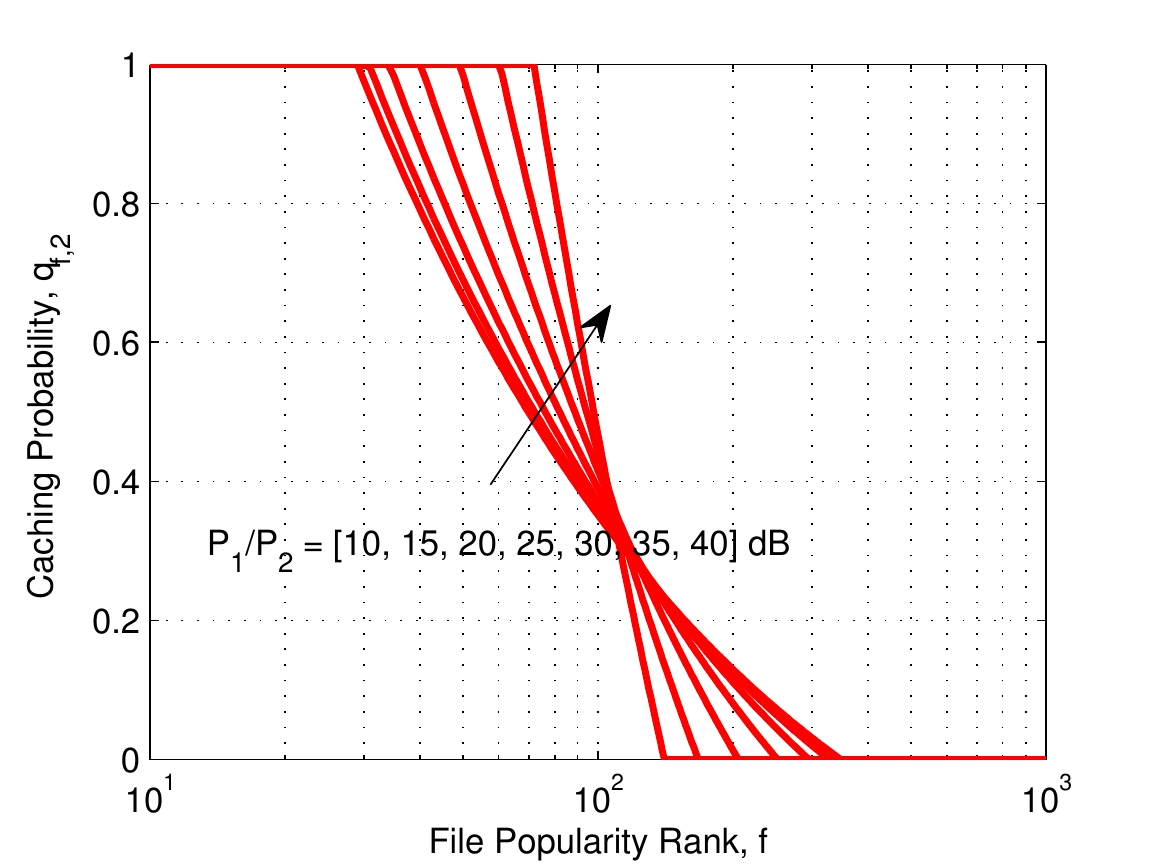}
	\vspace{-2mm}
	\caption{Impact of transmit power on optimal caching probabilities, $\delta = 1$.}
	\label{fig:P} 	
		
\end{figure}

In Fig. \ref{fig:P}, we show the impact of transmit power on the optimal caching policy. When $P_{1}/P_2$ increases, the files with higher popularity have more chances to be cached while the files with lower popularity have less chances to be cached, which agrees with Corollary 2. This is because the interference from the MBS increases with $P_1/P_2$ and caching the most popular file everywhere can increase the user's SINR, which leads to the increase of successful offloading probability.

\begin{figure}[!htb]
	\centering 		
	\includegraphics[width = 0.4\textwidth]{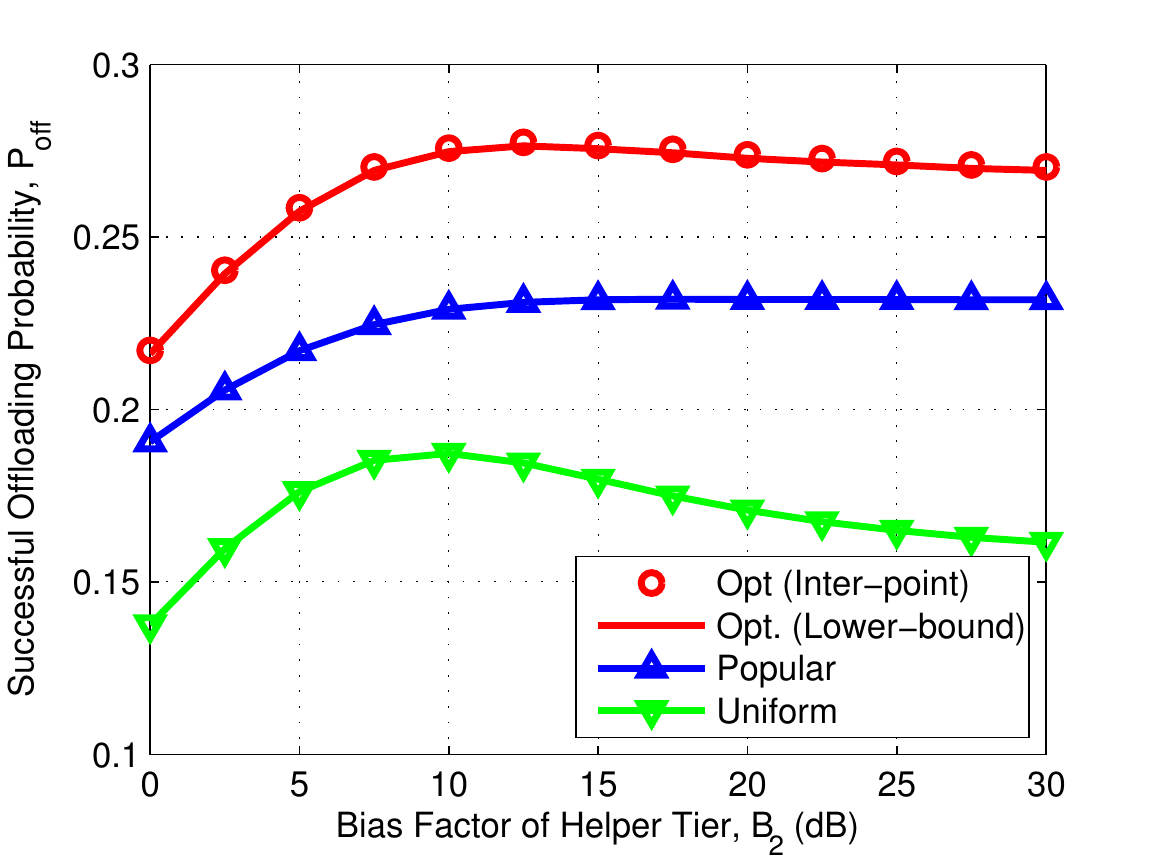}
				\vspace{-3mm}
	\caption{Impact of bias factor on successful offloading probability, $\delta = 0.5$.}
	\label{fig:B}
\end{figure}

In Fig. \ref{fig:B}, we show the impact of the bias factor. We can see that the optimized caching policies outperform the other two policies. When $B_2$ increases, the successful offloading probability first increases and then decreases (though slightly for ``Popular"). This is because the users are more likely to be offloaded to the helper tier when $B_2$ increases, and hence increases the successful offloading probability. Meanwhile, the distance between the user and its interfering MBS increases since the user prefers a far helper node to associate with than a near MBS, which reduces the SINR of user and may decrease the successful probability. 
\vspace{-1mm}
\section{Conclusion}
In this paper, we investigated optimal content placement in cache-enabled HetNets. We derived the closed-form expression of the successful offloading probability and obtained the optimal caching probability maximizing the successful offloading probability. We then analyzed the impact of BS density, user density, transmit power, and SINR threshold on the optimal caching policy. Simulation and numerical results validated our analysis and showed that when the ratios of MBS-to-helper density, MBS-to-helper transmit power, and user-to-helper density and the SINR threshold are large, the optimal caching policy tends to cache the most popular files everywhere. Besides, there exists optimal bias factor to maximize the successful offloading probability.
\appendices
\section{Proof of Proposition 1}
Based on the formula of total probability, the successful offloading probability defined in \eqref{eqn:def} can be derived as \vspace{-1mm}
\begin{equation}
P_{\rm off} = \sum_{f=1}^{N_f} p_f  \mathbb{P} (\gamma> \gamma_0 , k=2 | f) 
 =  \sum_{f=1}^{N_f} p_f \mathcal{P}_{f,2} \mathbb{P} (\gamma_{f,2} > \gamma_0)  \label{eqn:off}
\end{equation}
where $\mathbb{P}(\gamma > \gamma_0, k =2 | f)$ is the successful offloading probability conditioned on that the typical user requests the $f$th file, $\mathcal{P}_{f,k} = q_{f,k}\big(\sum_{j=1}^{2}q_{f,j}\lambda_{jk} (P_{jk}B_{jk})^{\frac{2}{\alpha}} \big)^{-1}$ \cite{flexible} is the probability that the typical user is associated with the $k$th tier when requesting the $f$th file, and $\mathbb{P}(\gamma_{f,k} > \gamma_0)$ is the successful transmission probability when the typical user requests the $f$th file and associates with the $k$th tier.

Based on the formula of total probability, we have
\begin{equation}
\mathbb P (\gamma_{f,k} > \gamma_0 ) = \int_{0}^{\infty}\mathbb{P} (\gamma_{f,k} > \gamma_0 | r) f_{r_k}(r) dr \label{eqn:suc3}
\end{equation}
where $f_{r_{k}}(r)  = \frac{2 \pi q_{f,k}\lambda_k}{\mathcal{P}_{f,k}} r  \exp (-\pi r^2 \sum_{j=1}^{2} q_{fj}\lambda_j ({P}_{jk}B_{jk})^{\frac{2}{\alpha}} )$ is the probability density function of the distance between user and its serving BS when requesting the $f$th file and associated with the $k$th tier, and $\mathbb P (\gamma_{f,k} > \gamma_0 | r )$ is the conditional successful transmission probability, which can be derived as,
\begin{align}
&\mathbb{P} (\gamma_{f,k} > \gamma_0 | r)
 \overset{(a)}{=} \mathbb{E}_{I_k} \left[  \mathbb{P}\left[ h_{k0} > M_k P_k^{-1} r^{\alpha} I_k \gamma_0 |r, I_k\right] \right] \nonumber \\
& \overset{(b)}{=} \mathbb{E}_{I_k} \left[\exp(-M_kP_k^{-1}r^\alpha I_k \gamma_0)\right] \nonumber\\
& \overset{(c)}{=} \prod_{j=1}^{2} \! \mathcal{L}_{I_{f,kj}}\! \left(M_kP_k^{-1}r^\alpha \gamma_0 \right)  \prod_{j=1}^{2} \mathcal{L}_{I_{f',kj}} \left(M_kP_k^{-1}r^\alpha \gamma_0 \right) \label{eqn:con_suc}
\end{align}
where step $(a)$ is from \eqref{eqn:gamma} by neglecting the thermal noise and using the law of total probability, step $(b)$ is from $h_{k0} \sim \exp(1)$, step $(c)$ follows because $\mathcal{L}_{\sum_{j}I_{f,kj}}(s)
=\prod_{i} \mathcal{L}_{I_{f, kj}}(s)$, and $\mathcal{L}(\cdot)$ denotes the Laplace
transform.

The Laplace transform of $I_{f,kj}$  can be derived as
\begin{align}
&   \mathcal{L}_{I_{f,kj}} (s) = \mathbb{E}_{\tilde \Phi_{f,j}, h_{ji}} \Big[e^{-s \sum_{i\in \tilde\Phi_{f,j} \backslash b_{k0}} P_j  h_{ji} r_{ji}^{-\alpha_j}}\Big] \nonumber \\
&  \overset{(a)}{=}  \mathbb{E}_{\tilde \Phi_{f,j}} \Big[\!\!\! \prod_{i\in \tilde\Phi_{f,j} \backslash b_{k0}}\!\!\! \!\!\!\left(1 + \tfrac{s P_j}{M_j}  r_{j,i}^{-\alpha_j}  \right)^{-M_j}\Big]  \nonumber \\
&   \overset{(b)}{=}\! e^{ -2\pi p_{{\rm a}, j} q_{f,j}\lambda_j \int_{r_{0j}}^{\infty} \big(1 - ({1 + \frac{s P_j}{M_j} u^{-\alpha_j}} )^{-M_j}\big)u du} \nonumber\\
&  = e^{ - \pi p_{{\rm a},j} q_{f,j} \lambda_j r_{0j}^2 \left(F_1 [ -\frac{2}{\alpha}, M_j; 1-\frac{2}{\alpha}; -\frac{sP_j}{M_j} r_{0j}^{-\alpha}] -1\right)  } \label{eqn:Laplace 1}
\end{align}
where step $(a)$ follows from $h_{ji}\sim \mathbb G(M_j,1/M_j)$, step $(b)$ comes from approximating the distribution of active BSs caching the $f$th file $\tilde{\Phi}_{f,j}$ as PPP with density $p_{a,j}q_{f,j}\lambda_j$ \cite{economy} and then using the probability generating function of the PPP, $p_{{\rm a}, j} \approx 1 - \big(1 + \frac{\mathcal{P}_k \lambda_u}{3.5\lambda_k}\big){}^{-3.5}$ is the probability that a randomly chosen BS in the $j$th tier is active \cite{offloading}, $\mathcal{P}_k = \sum_{f=1}^{N_f} p_f \mathcal{P}_{f,k} = \sum_{f=1}^{N_f} p_f  q_{f,k}\big(\sum_{j=1}^{2}q_{f,j}\lambda_{jk} (P_{jk}B_{jk})^{\frac{2}{\alpha}} \big)^{-1}$ is  the probability that the typical user is associated with the $k$th tier, and $r_{0j} = (P_{jk}B_{jk})^{\frac{1}{\alpha}} r$ is the closest possible distance of the interfering BS in $\tilde \Phi_{f,j}$.
 Substituting $r_{0j}$ and $s = M_k P_k^{-1} r^\alpha \gamma_0$ into \eqref{eqn:Laplace 1}, we obtain
\begin{equation}
\mathcal{L}_{f,kj}(M_k P_k^{-1} r^\alpha \gamma_0)
 = e^{  - \pi p_{{\rm a},j} q_{f,j}\lambda_j  (P_{jk}B_{jk})^{\frac{2}{\alpha}} r^2 \mathcal{Z}_{f,kj} (\gamma_0) } \hspace{-2.5mm}\label{eqn:L1}
\end{equation}
where $ \mathcal{Z}_{f,kj} (\gamma_0)\triangleq {}_{2}F_1 \big[ -\frac{2}{\alpha}, M_j; 1-\frac{2}{\alpha}; -\frac{\gamma_0}{M_{jk}B_{jk}} \big] -1$.

Since the BSs not caching the $f$th file, i.e. $\tilde{\Phi}_{f',j}$, can be arbitrarily close to the user, i.e., $r_{0j} \to 0$, from $\lim\limits_{r_{0j} \to 0}  r_{0j}^2(_2F_1 \big[-\frac{2}{\alpha}, M_j; 1-\frac{2}{\alpha}; -\frac{sP_j}{M_j} r_{0j}^{-\alpha}\big] -1 ) = \Gamma (1-\frac{2}{\alpha}) \Gamma(M_j + \frac{2}{\alpha_j}) \Gamma(M_j)^{-1}(\frac{sP_j}{M_j})^{\frac{2}{\alpha}}$, similar to the derivation of \eqref{eqn:Laplace 1}, we can obtain
\begin{equation}
\mathcal{L}_{f',kj}(M_k P_k^{-1} r^\alpha \gamma_0)= e^{-\pi p_{{\rm a},j} (1-q_{f,j})\lambda_j P_{jk}^{\frac{2}{\alpha}} r^2 \mathcal{Z}_{f',kj}(\gamma_0) } \hspace{-1mm} \label{eqn:L2}
\end{equation}
where $\mathcal{Z}_{f',kj}(\gamma_0) \triangleq \Gamma \left(1-\frac{2}{\alpha}\right)  \Gamma\left(M_j + \frac{2}{\alpha}\right) \Gamma(M_j)^{-1}(\frac{\gamma_0}{M_{jk}})^{\frac{2}{\alpha}}$.

By substituting \eqref{eqn:L1} and \eqref{eqn:L2} into \eqref{eqn:con_suc} and then \eqref{eqn:con_suc} into \eqref{eqn:suc3}, we obtain
	\begin{multline}
	\mathbb{P} ( \gamma_{f,k} > \gamma_0 )  = \frac{q_{f,k}}{\mathcal{P}_{f,k}}  2\bigg(\sum_{j=1}^{2}  \lambda_{jk} P_{jk}^{\frac{2}{\alpha}} \Big(p_{{\rm a},j} \mathcal{Z}_{f',jk} (\gamma_0)+q_{fj}   \\
	   \times p_{{\rm a},j}\Big( B_{jk}^{\frac{2}{\alpha}} \mathcal{Z}_{f,jk}(\gamma_0) -  \mathcal{Z}_{f',jk}(\gamma_0)\Big)  +  q_{fj}   B_{jk}^{\frac{2}{\alpha}}    \Big) \bigg)^{-1}\label{eqn:suc}
	\end{multline}
Substituting $k=2$, $q_{f,1} = 1$, $p_{{\rm a}, 1} = 1$ (since $\lambda_u \gg \lambda_1$) into \eqref{eqn:suc} and then into \eqref{eqn:off}, Proposition 1 can be proved.

\section{Proof of Corollary 1}
Without loss of generality, we assume $q_{1,2}^*, \dots, q_{N_1,2}^* = 1$ and $q_{N_f-N_0+1,2}^*, \dots, q_{N_f,2}^* = 0$. From $\sum_{f=1}^{N_f} q_{f,2}^* = N_c$ and by defining $c_1 \triangleq \mathcal{C}_{1, \gamma_0} + \mathcal{C}_{2, \gamma_0}$ and $c_2 \triangleq C_{3,\gamma_0} + 1$, we have 
$
\sum_{f=N_0 + 1}^{N_f-N_1} \left( \frac{1}{c_2} \sqrt{\frac{c_1}{\nu}} \sqrt{p_f} - \frac{c_1}{c_2}\right) + N_1 = N_c
$,
from which we can obtain
\begin{equation}
 \frac{1}{c_2}\sqrt{\frac{c_1}{\nu}}  = \frac{N_c - N_1 + (N_f - N_0 - N_1)\frac{c_1}{c_2}}{ \sum_{f=N_0 + 1}^{N_f-N_1} \sqrt{p_f}} \label{eqn:k}
\end{equation}
We can see that $ \frac{1}{c_2}\sqrt{\frac{c_1}{\nu}}$ increases with $c_1$ since $c_2 \geq 0$.\footnote{It can be proved that $-1 \leq \mathcal{C}_{3,\gamma_0} \leq 0$ for $\gamma_0 \in [0,\infty)$, and hence $0 \leq c_2 \leq 1$.} Considering that $c_1$ increases with but $c_2$ does not depend on $\lambda_{12}$ and $P_{12}$, $\frac{1}{c_2}\sqrt{\frac{c_1}{\nu}}$ increases with $\lambda_{12}$ and $P_{12}$.

For any $q_{f,2}^*$, $q_{f+1,2}^* \in (0,1)$, from \eqref{eqn:opt}, we have
\begin{equation}
 q_{f,2}^* - q_{f+1,2}^* =  \frac{1}{c_2}\sqrt{\frac{c_1}{\nu}} (p_f - p_{f+1}) \label{eqn:diff}
\end{equation}
Since $p_f > p_{f+1}$, $q_{f,2}^*-q_{f+1,2}^* $ increases with $\frac{1}{c_2}\sqrt{\frac{c_1}{\nu}}$ and hence increases with $\lambda_{12}$ and $P_{12}$.

\section{Proof of Corollary 2}
By using the transformation of  ${}_2 F_1 [a,b;c;z]  $ \cite[eq. (9.132)]{jeffrey2007table}, and considering the series-form expression of ${}_2F_1 [a,b;c;z] = \sum_{n = 0}^{\infty} \frac{(a)_n 	(b)_n}{(c)_n} z^n$ where $(x)_n \triangleq x(x+1) \cdots (x + n - 1)$
denotes the rising Pochhammer symbol, we can obtain the asymptotic result of $ {}_{2}F_1  \big[ -\frac{2}{\alpha}, M_2; 1-\frac{2}{\alpha}; -\gamma_0 \big]$ for $\gamma_0 \to \infty$ as $\Gamma (1-\frac{2}{\alpha}) \Gamma(M_2 + \frac{2}{\alpha_2}) \Gamma(M_2)^{-1}\gamma_0{}^{\frac{2}{\alpha}} $ which equals $\mathcal{C}_{2,\gamma_0}$. Then, considering the definition of $\mathcal{C}_{3, \gamma_0}$, we have $\mathcal{C}_{3,\gamma_0} \triangleq {}_{2}F_1  \big[ -\frac{2}{\alpha}, M_2; 1-\frac{2}{\alpha}; -\gamma_0 \big] - \mathcal C_{2,\gamma_0}- 1 = -1$ for $\gamma_0 \to \infty$. Upon substituting $\mathcal{C}_{3,\gamma_0} = -1$ into \eqref{eqn:case1}, we obtain $P_{\rm off}^{\infty} = \frac{1}{\mathcal{C}_1(\gamma_0) + \mathcal{Z}_{f',22}}\sum_{f=1}^{N_f} p_fq_{f,2}$ for $\gamma_0 \to \infty$. Since $p_f$ decreases with $f$ and further considering constraint \eqref{eqn:con1} and \eqref{eqn:con2}, it is easy to see that the optimal values of $q_{f,2}$ maximizing $\sum_{f = 1}^{N_f} p_f q_{f,2}$ are $q_{1,2}^*, \cdots, q_{N_c,2}^* = 1$ and $q_{N_c+1,2}^*, \cdots, q_{N_f, 2}^* = 0$, and hence Corollary 2 is proved.

\section{Proof of Corollary 3}
Since \eqref{eqn:optlower} has the same structure as \eqref{eqn:opt}, similar to the proof in Appendix B, we can assume $\ubar q_{1,2}^*, \dots, \ubar q_{N_1,2}^*  = 1$ and $\ubar q_{N_f-N_0+1,2}^*, \dots, \ubar q_{N_f,2}^*   =0$ without loss of generality.

By defining $\bar c_1 \triangleq  \mathcal{C}_{1,\gamma_0}  + \bar{p}_{\rm a, 2}\mathcal{C}_{2,\gamma_0}$ and $\bar c_2 \triangleq \bar p_{\rm a, 2}C_{3,\gamma_0} + 1$, similar to the derivation of \eqref{eqn:k}, we can obtain
$
\frac{1}{\bar c_2}\big(\frac{\bar c_1}{\nu}\big){}^{\frac{1}{2}}  = \big({N_c - N_1 + (N_f - N_0 - N_1)\frac{\bar c_1}{\bar c_2}}\big){ / \big(\sum_{f=N_0 + 1}^{N_f-N_1} \sqrt{p_f}\big)} \label{eqn:}
$.
Since  $\mathcal{C}_{2,\gamma_0} \geq 0$ and $\mathcal{C}_{3,\gamma_0} \leq 0$, ${\bar c_1}/{\bar c_2}$ increases with $\bar p_{{\rm a}, 2}$ and hence $\frac{1}{\bar c_2}\big(\frac{\bar c_1}{\nu}\big){}^{\frac{1}{2}}$ increases with $\bar p_{{\rm a}, 2}$. Further considering that $p_{{\rm a},2}$ increases with $\lambda_u/\lambda_2$, we know that $\frac{1}{\bar c_2}\big(\frac{\bar c_1}{\nu}\big){}^{\frac{1}{2}}$ increases with  $\lambda_u/\lambda_2$. Then, with the similar way to derive \eqref{eqn:diff}, Corollary 3 can be proved.

\bibliographystyle{IEEEtran}
\bibliography{dongbib}

\end{document}